\documentclass[a4paper,12pt,twoside]{article}
\usepackage{amsmath, amsfonts, amssymb, amscd}
\usepackage{url}
\usepackage[pdfborder={0 0 0}]{hyperref}
\usepackage{graphicx}
\usepackage{mathrsfs}
\usepackage{amsthm}
\usepackage{braket}
\usepackage{accents}
\usepackage{oldgerm}
\usepackage{color}
\usepackage{hyperref}
\usepackage{yfonts}
\usepackage{ctable}

\newlength{\dinwidth}
\newlength{\dinmargin}

\newcommand{\vphi}{\varphi}

\DeclareMathAlphabet{\mathpzc}{OT1}{pzc}{m}{it}

\theoremstyle{plain}
\newtheorem{theorem}{Theorem}[section]
\newtheorem{proposition}[theorem]{Proposition}
\newtheorem{lemma}[theorem]{Lemma}

\theoremstyle{definition}
\newtheorem{definition}[theorem]{Definition}
\newtheorem{remark}[theorem]{Remark}

\theoremstyle{remark}

\numberwithin{equation}{section}

\newcommand{\Ibb}[1]{ {\rm I\ifmmode\mkern -3.6mu\else\kern -.2em\fi#1}}
\newcommand{\ibb}[1]{\leavevmode\hbox{\kern.3em\vrule
     height 1.2ex depth -.3ex width .2pt\kern-.3em\rm#1}}
\newcommand{\Cl}{{\ibb C}}
\newcommand{\Rl}{{\Ibb R}}
\newcommand{\Nl}{{\Ibb N}}
\newcommand{\Zl}{\mathbb{Z}}
\newcommand{\Quat}{\mathbb{H}}

\definecolor{lightgray}{rgb}{0.8,0.8,0.8}

\newcommand{\ra}{\rightarrow}
\newcommand{\lra}{\longrightarrow}

\newcommand{\uhr}{\upharpoonright}

\newcommand{\Bl}{\biggl}
\newcommand{\Br}{\biggr}
\newcommand{\bl}{\bigl}
\newcommand{\br}{\bigr}

\newcommand{\Aut}{\text{\normalfont\textrm{Aut}}}

\newcommand{\id}{\mathrm{id}}

\newcommand{\eps}{\varepsilon}

\newcommand{\bs}{\boldsymbol}
\renewcommand{\cal}{\mathcal}
\newcommand{\scr}{\mathscr}

\newcommand{\cc}{\scr{C}}

\newcommand{\Tr}{\mathrm{Tr}}

\newcommand{\BH}{\cal{B(H)}}

\newcommand{\bj}{{\bs j}}

\newcommand{\gdS}{\mathrm{\bf g}}

 \newcommand{\CAR}{\mathrm{CAR}}

\newcommand{\del}{\partial}

\newcommand{\Mat}{\mathrm{Mat}}

\newcommand{\F}{\cal{F}}

\newcommand{\SO}{\mathrm{SO}}

\newcommand{\Uone}{\mathrm{U(1)}}
\newcommand{\Spin}{\mathrm{Spin}}
\newcommand{\Sp}{\mathrm{Sp}}

\renewcommand{\L}{\cal{L}}
\newcommand{\tLo}{\widetilde{\L}_0}

\newcommand{\veps}{\varepsilon}

\newcommand{\A}{\cal{A}}

\newcommand{\HS}{\cal{H}}

\newcommand{\hs}{\scr{H}}

\newcommand{\OO}{\mathcal{O}}

\newcommand{\W}{\cal{W}}

\newcommand{\U}{\mathrm{{\bf U}}}

\newcommand{\C}{\cal{K}}
\renewcommand{\S}{\cal{S}}

\allowdisplaybreaks
\begin{document} 
%
%
\begin{center}
	$ $

	\vspace{3cm}
	{\large\bf Deformations of quantum field theories on de Sitter spacetime} 
	
	\vspace{0.6cm}  
	\noindent {Eric Morfa-Morales}
	
	\vspace{0.2cm}
	\noindent {Mathematical Physics Group, University of Vienna,}
	
	{Boltzmanngasse 5, A-1090, Vienna, Austria}

	\texttt{eric.morfa.morales@univie.ac.at}

\end{center}

\vspace{0.5cm}\noindent
{\bf Abstract:} Quantum field theories on de Sitter spacetime with global $\Uone$ gauge symmetry are deformed using the joint action of the internal symmetry group and a one-parameter group of boosts. The resulting theory turns out to be wedge-local and non-isomorphic to the initial one for a class of theories, including the free charged Dirac field. The properties of deformed models coming from inclusions of $\CAR$-algebras are studied in detail.

%
\newpage
\tableofcontents

\section{Introduction} 
\label{sec:}
The construction of interacting quantum field theories by deformations using suitable actions of $\Rl^n$ has recently attracted much attention \cite{GrosseLechner07}, \cite{BuchholzSummers08}, \cite{GrosseLechner08}, \cite{BuchholzLechnerSummers10}, \cite{DybalskiTanimoto10}, \cite{DappiaggiLechnerMorfa-Morales10}. A first example of this kind was found for the scalar free field on Minkowski space by Grosse and Lechner \cite{GrosseLechner07}, where the deformed operators fulfill a weakened form of locality (wedge-locality) and the two-particle scattering is non-trivial. Later on this deformation method (warped convolution) was generalized to arbitrary quantum field theories on Minkowski space by Buchholz and Summers \cite{BuchholzSummers08}. It is formulated in terms of an action of the translation subgroup of the Poincar\'e group and the resulting theory has similar properties as in the scalar free field case. Within a Wightman setting this deformation manifests itself as a deformation of the tensor product of the underlying Borchers-Uhlmann algebra \cite{GrosseLechner08}. Subsequently \cite{BuchholzLechnerSummers10} it was realized that there is in fact a close connection between warped convolutions and a well known deformation method for C$^*$-algebras in mathematics, namely, Rieffel deformations \cite{Rieffel93}. It turns out that the warped algebra forms a representation of the Rieffel deformed algebra for a fixed deformation parameter. Later on this deformation scheme was applied to various situations in quantum field theory. In the chiral conformal case the first examples of massless models which are interacting and asymptotically complete were constructed by these methods \cite{DybalskiTanimoto10}. Quantum field theories on a class of curved spacetimes relevant to cosmology can also be deformed using the flow of suitable Killing vector fields instead of translations \cite{DappiaggiLechnerMorfa-Morales10}.

The de Sitter spacetime does not belong to the class of spacetimes considered in \cite{DappiaggiLechnerMorfa-Morales10}. It is the purpose of this paper to apply the warped convolution deformation procedure to quantum field theories on this spacetime. 
We use a combination of external and internal symmetries, consisting of a one-parameter group of boosts associated with a wedge and a global $\Uone$ gauge symmetry, as an $\Rl^2$-action to define the deformation. The resulting theory is wedge-local and unitarily inequivalent to the undeformed one for a class of theories, including the free charged Dirac field. 

In Section \ref{sec:deSitterSpacetime} the basic notions concerning the geometry and causal structure of de Sitter space are recalled and we discuss the de Sitter group together with its universal covering. After that, the covariance and inclusion properties of a class of distinguished regions (wedges) in de Sitter space are studied. 

In Section \ref{sec:DeformationsQFT} we consider quantum field theories with global gauge symmetry within the algebraic setting (field nets) and we show how to reconstruct a wedge-local field net from an inclusion of two C$^*$-algebras, which are in a suitable relative position to a wedge. Then the warped convolution deformation is applied to a field net with global $\Uone$ gauge symmetry  and the properties of the resulting theory are studied. 

In Section \ref{subsec:CARnets} a particular class of field nets is investigated in more detail, namely, nets of $\CAR$-algebras. For these theories the deformed operators can be computed explicitly. The fixed-points of the deformation map are determined and it is shown that the deformed and undeformed field nets are non-isomorphic.

In the conclusions we comment on warped convolutions in terms of purely external and internal symmetries using other Abelian subgroups of the de Sitter and gauge group.

\section{de Sitter spacetime} 
\label{sec:deSitterSpacetime}
\subsection{Geometry and causal structure} 
\label{subsec:}
The de Sitter spacetime $(M,\gdS)$ is a vacuum solution of Einstein's equation with positive cosmological constant. 
It is maximally symmetric, so it admits 10 Killing vector fields, which is the maximum number for a spacetime of dimension four. It is also globally hyperbolic, so the Cauchy problem for partial differential equations of hyperbolic type, such as the Klein-Gordon and Dirac equation, is well-posed. Furthermore, it is a special case of the Friedmann-Robertson-Walker spacetimes which describe a spatially homogeneous and isotropic universe and it plays a prominent role in many inflationary scenarios for the early universe \cite{Linde09}.

Most conveniently it can be represented as the embedded submanifold
\begin{equation*}
	M=\bl\{x\in\Rl^5: \eta(x,x)=-1\br\} 
\end{equation*}
of five-dimensional Minkowski space $(\Rl^5,\eta)$, where $\Rl^5$ is identified with $T_x\Rl^5$, $x\in\Rl^5$. The signature of $\eta$ is $(1,-1,-1,-1,-1)$ and the de Sitter radius is fixed to one. The metric $\gdS$ on $M$ is the induced metric from the ambient space, {\it i.e.}, $\gdS=\iota^*\eta$, where $\iota:M\hookrightarrow \Rl^5$ is the embedding map. We use the ambient space notation to parametrize the de Sitter hyperboloid, so we write $x=(x^0,x^1,\vec{x})$, $\vec{x}=(x^2,x^3,x^4)$ for points in $M$, subject to the relation $(x^0)^2-\sum_{k=1}^4(x^k)^2=-1$, where $\{x^\mu:\mu=0,\dots 4\}$ is a Cartesian coordinate system of $\Rl^5$.

Since the metric on $M$ is the induced metric from the ambient Minkowski space, the causal structure is also inherited. Hence points in $(M,\gdS)$ are called timelike, spacelike or null related, if they are so as points in $(\Rl^5,\eta)$, respectively. We fix a time orientation in $(M,\gdS)$ once and for all. The interior of the causal complement of a spacetime region $\OO\subset M$ is denoted by $\OO'$.

For the generators of the Clifford algebra which is associated with the quadratic form $\eta(x,x)$ on the vector space $\Rl^5$ we use the representation \cite{Gazeau07}
\begin{equation*}
\gamma_0=
\begin{pmatrix}
 1 & 0\\
 0 & -1
\end{pmatrix},\qquad
\gamma_1=
\begin{pmatrix}
 0 & 1\\
 -1 & 0
\end{pmatrix},\qquad
\gamma_k=
\begin{pmatrix}
 0 & e_k\\
 e_k & 0
\end{pmatrix},\, k=2,3,4,
\end{equation*}
where $e_k=(-1)^k\sigma_{k-1}$ and $\sigma_{k-1}$ are the Pauli matrices. We have $\{\gamma_\mu,\gamma_\nu\}=2\eta_{\mu\nu}\cdot 1$ and $\{1,e_1,e_2,e_3\}$ is a basis of the quaternions $\Quat$. Note that this representation is in fact not faithful, since $i\gamma_0\cdots\gamma_4=1$. Similar to the case of four-dimensional Minkowski space, where points are parametrized by hermitian $2\times 2$ matrices, we parametrize points on the de Sitter hyperboloid by $2\times 2$ quaternionic matrices. This parametrization is useful for the discussion of the covering group of the de Sitter group later on. Define
\begin{equation*}
 M\ni (x^0,x^1,\vec{x})\longmapsto
 \undertilde{x}:=\sum_{\mu=0}^4x^\mu\gamma_\mu=
\begin{pmatrix}
 x^0 & -q \\
 \overline{q} & -x^0
\end{pmatrix}
\in\Mat(2,\Quat),
\end{equation*}
where $\overline{q}=(x^1,-\vec{x})$ is the quaternionic conjugate of $q=(x^1,\vec{x})$. Conversely, every $2\times 2$ matrix of the above form determines a point in de Sitter space via
\begin{equation*}
x^\mu=\frac{1}{4}\Tr(\gamma_\mu \undertilde{x}).
\end{equation*}
The map $x\mapsto \undertilde{x}$ defines an isomorphism between $M$ and $H(2,\Quat)\gamma_0\subset\Mat(2,\Quat)$, where $H(2,\Quat)$ are the hermitian $2\times 2$ matrices over $\Quat$. Furthermore, there holds $\eta(x,x)1= \undertilde{x}^*\gamma_0\undertilde{x}\gamma_0$, where $\undertilde{x}^*$ is the transpose of the quaternionic conjugate of $\undertilde{x}$.

\subsection{The de Sitter group and its covering}
\label{subsec:DeSitterGroup}
The isometry group of $(M,\gdS)$ is 
\begin{equation*}
 \mathrm{O}(1,4)=\{\Lambda\in\Mat(5,\Rl):\Lambda^T\eta\Lambda=\eta\}
\end{equation*}
and its action on $M$ is given by the action of the Lorentz group in the ambient Minkowski space.
This group is a ten-dimensional, non-compact, non-connected and real Lie group which has four connected components. The connected component which contains the identity is denoted by $\L_0:=\SO(1,4)_0$. This group is called de Sitter group (proper orthochronous Lorentz group) and its elements preserve the orientation and time orientation of $(M,\gdS)$. 

Since we also want to treat quantum fields with half-integer spin we consider the two-fold (and universal) covering of $\L_0$, which is the spin group $\widetilde{\L}_0:=\Spin(1,4)$. Hence there exists a short exact sequence of group homomorphisms
\begin{equation*}
 1\lra\ker(\pi)=\{\pm 1\}\lra \widetilde{\L}_0\overset{\pi}{\lra} \L_0\lra 1.
\end{equation*}
There holds $\L_0\cong\widetilde{\L}_0/\{\pm 1\}$ and $\widetilde{\L}_0$ is simply connected. Note that the Lie group $\Spin(1,4)$ is isomorphic to the pseudo-symplectic group \cite{Takahashi63}
\begin{equation*}
	\Sp(1,1)=\Bl\{
\begin{pmatrix}
	a & b\\
	c & d 
\end{pmatrix}
\in \Mat(2,\Quat):\;\bar{a}b=\bar{c}d,\; |a|^2-|c|^2=1,\; |d|^2-|b|^2=1\Br\}.
\end{equation*}
Equivalently, $g\in\Sp(1,1)$ if and only if $g^*\gamma_0g=\gamma_0$. In this representation the covering homomorphism $\pi:\widetilde{\L}_0\ra \L_0$ is given by
\begin{equation*}
 (\pi(g))_{\mu\nu}=\frac{1}{4}\Tr(\gamma_\mu g\gamma_\nu g^{-1}),\quad g\in\tLo
\end{equation*}
and $\tLo$ acts on $M$ by conjugation $\undertilde{x}\mapsto g\undertilde{x}g^{-1}$.

\subsection{de Sitter wedges} 
\label{subsec:wedges}
Now we discuss the typical localization regions of the deformed quantum fields from Section \ref{sec:DeformationsQFT}. In \cite{BorchersBuchholz99} a de Sitter wedge is defined as the causal completion of the worldline of a uniformly accelerated observer (timelike geodesic) in de Sitter space. Equivalently, they can be characterized as intersections of wedges in the ambient Minkowski space \cite{ThomasWichmann97} and the de Sitter hyperboloid. Hence we specify a reference (or right) wedge by
\begin{equation*}
 W_0:=\{x\in\Rl^5: x^1>|x^0|\}\cap M
\end{equation*}
and define the family of wedges $\W$ as the set of all de Sitter transforms of $W_0$:
\begin{equation*}
 \W:=\{g W_0:g\in\L_0\}.
\end{equation*}
By definition, $\L_0$ acts transitively on $\W$. Each wedge $W\in\W$ has an attached edge $E_W$ which is a two-sphere. We have $E_{W_0}=\{x\in M:x^0=x^1=0\}$ and $E_W=g E_{W_0}$ for $W=g W_0$. The wedge $W$ coincides with a connected component of the causal complement of the edge $E_W$ \cite{BorchersBuchholz99}. For the stabilizer of the wedge $W$ we write 
$\L_0(W):=\{g\in\L_0:g W=W\}$.

From the properties of wedges in $(\Rl^5,\eta)$ follows that the causal complement of a wedge is again a wedge and that every $W\in\W$ is causally complete, {\it i.e.}, $W''=(W')'=W$.  
Furthermore, the family $\W$ is causally separating, so given spacelike separated double cones $\OO_1,\OO_2\subset M$, there exists a $W\in\W$ such that $\OO_1\subset W\subset \OO_2\,'$ (see \cite{ThomasWichmann97}). 
\begin{remark}
 Wedges are frequently used as localization regions in quantum field theory \cite{BisognanoWichmann75}, \cite{Borchers00}, \cite{BuchholzDreyerFlorigSummers00}.  Although a net of observables over wedges is to a certain degree non-local, it is possible to construct a local theory over double cones from it. This is achieved by taking suitable intersections of algebras associated with wedges. Since $\W$ is causally separating the resulting theory satisfies, in particular, local commutativity (see \cite{BuchholzSummers08,BuchholzLechnerSummers10}). To decide whether these intersections are in fact non-trivial remains a challenging task in four dimensions.
\end{remark}

For every $W\in\W$ there exists a one-parameter group $\Gamma_W=\{\Lambda_W(t)\in\L_0:t\in\Rl\}$, such that each $\Lambda_W(t),t\in\Rl$ maps $\W$ onto $\W$ and $\Lambda_W(t) W=W$ for all $t\in\Rl$. Moreover 
\begin{equation}
\label{eq:GammaCovarianceL0}
 \Lambda_{gW}(t)=g\Lambda_W(t) g^{-1},\quad g\in\L_0,\; t\in\Rl.
\end{equation}
Associated with $\Gamma_W$ is a future-directed Killing vector field $\xi_W$ in the wedge $W$ and the worldline from which the wedge is constructed is an integral curve of (a portion of) this vector field. 
Furthermore, for every $W\in\W$ there exists a reflection $j_W\in \L_0$ which maps $\W$ onto $\W$ and satisfies 
\begin{equation}
\label{eq:WedgeReflection}
 j_W W=W',\quad j_{gW}=gj_Wg^{-1},\quad g\in\L_0.
\end{equation}
Since $\L_0$ acts transitively on $\W$ we only need to specify these maps for $W_0$. We choose
\begin{equation}
\label{eq:GammaW0}
	\Lambda_{W_0}(t):=\left( 
	\begin{array}{ccc}
		\cosh(2\pi t) & \sinh(2\pi t) & 0\\
		\sinh(2\pi t) & \cosh(2\pi t) & 0\\
		0&0& 1_3 
	\end{array}
	\right),\quad j_{W_0}(x^0,x^1,\vec{x}):=(x^0,-x^1,-\vec{x})
\end{equation}
and note that $\xi_{W_0}=x^1\del_{x^0}+x^0\del_{x^1}$ is the associated Killing vector field.
\begin{remark}
 Within the context of applications of Tomita-Takesaki modular theory in quantum field theory the standard choice for the reflection  is $(x^0,x^1,\vec{x})\mapsto (-x^0,-x^1,\vec{x})$, which is an element of the extended symmetry group $\L_+=\L_0\rtimes \Zl_2$. In this paper we have no intention to use these techniques and the choice (\ref{eq:GammaW0}) appears to be more natural since we restrict our considerations to $\L_0$. However, all of our results can be generalized to the group $\L_+$ in a straightforward manner.
\end{remark}

The following lemma collects the basic properties of these maps.

\begin{lemma}
\label{lem:Gamma-J}
 Let $W\in\W$ and $\Lambda_W(t)\in \Gamma_W$, $j_W$ be as above. Then
\begin{enumerate}
 \item[a)] $g\Lambda_W(t)g^{-1}=\Lambda_W(t),\; g\in\L_0(W),t\in\Rl$,
 \item[b)] $j_W\Lambda_W(t)j_W=\Lambda_W(-t),\; t\in\Rl$,
\end{enumerate}
\end{lemma}
\begin{proof}
	a): For $g\in\L_0(W)$ holds $\Lambda_W(t)=\Lambda_{gW}(t)=g\Lambda_W(t)g^{-1}$ for all
	$t\in\Rl$ by (\ref{eq:GammaCovarianceL0}). b) follows from $j_{W_0}\Lambda_{W_0}(t)j_{W_0}=\Lambda_{W_0}(-t)$ and (\ref{eq:GammaCovarianceL0}), (\ref{eq:WedgeReflection}).
\end{proof}

The stabilizer of $W$ has the form $\L_0(W)=\Gamma_W\times \SO(3)$, where $\SO(3)$ are rotations in $E_W$. Hence $\Gamma_W$ coincides with the center of $\L_0(W)$. From b) follows that the Killing vector fields associated with $W$ and $W'$ differ only by temporal orientation.

The following lemma shows that the possible causal configurations of wedges are very much constrained in de Sitter space.

\begin{lemma}
\label{lem:WedgeInclusions}
Let $W_1,W_2\in\W$ and $W_1\subset W_2$. Then $W_1=W_2$.
\end{lemma}
\begin{proof}
The wedges $W_1,W_2$ can be written as $W_k=M\cap \widetilde{W}_k$, $k=1,2$, where $\widetilde{W}_k$ is a wedge in the ambient Minkowski space. Since the causal closure of $W_k$ in $\Rl^5$ coincides with $\widetilde{W}_k$, there follows $\widetilde{W}_1\subset \widetilde{W}_2$ from $W_1\subset W_2$. As the edges $E_{\widetilde{W}_k}$ both contain the origin, there follows $E_{\widetilde{W}_1}=E_{\widetilde{W}_2}$ and also $E_{W_1}=E_{W_2}$ since $E_{W_k}=M\cap E_{\widetilde{W}_k}$. The assertion $W_1=W_2$ follows from the assumption $W_1\subset W_2$ together with the fact that $W_k$ is a connected component of the causal complement of $E_{W_k}$.
\end{proof}

\begin{remark}
All the previous statements carry over to the covering $\tLo$ in a straightforward manner. Define an action of $\tLo$ on $\W$ with the covering homomorphism
\begin{equation}
\label{eq:CoveringWedgeAction}
 g W:=\pi(g) W,\quad g\in\tLo,\, W\in\W,
\end{equation}
which is transitive, since $\L_0$ acts transitively. The one-parameter group $\Gamma_W\subset\L_0$ lifts to a {\em unique} one-parameter group $\widetilde{\Gamma}_W\subset\tLo$ and for its elements we write $\lambda_W(t),t\in\Rl$. Again, since $\tLo$ acts transitively on $\W$, we only need to specify these maps for $W_0$. We have \cite[p.368]{Takahashi63}
\begin{equation*}
 \lambda_{W_0}(t)=
\begin{pmatrix}
 \cosh(\pi t) & \sinh(\pi t)\\
 \sinh(\pi t) & \cosh(\pi t)
\end{pmatrix}.
\end{equation*}
Clearly, $\lambda_W(t)W=W$ and $\lambda_{gW}=g\lambda_{W}g^{-1}$ for all $g\in\tLo$, $W\in\W$ with respect to the action (\ref{eq:CoveringWedgeAction}). For the lift of the reflection $j_{W_0}$ we choose 
\begin{equation*}
 \bj_{W_0}:=
\begin{pmatrix}
 1 & 0\\
 0 & -1
\end{pmatrix}.
\end{equation*}                                                                                                                                            
Again, $\bj_W W=W'$ and $\bj_{gW}=g\bj_{W_0}g^{-1}$ for all $g\in\tLo$, $W\in\W$ with respect to (\ref{eq:CoveringWedgeAction}). Hence analogous statements as in Lemma \ref{lem:Gamma-J} hold for $\widetilde{\Gamma}_W$ with $\L_0$ replaced by $\tLo$, {\it i.e.},
\begin{equation}
\label{eq:CovarianceGammaJ}
	g\lambda_W(t)g^{-1}=\lambda_W(t),\qquad \bj_W\lambda_W(t)\bj_W=\lambda_W(-t),
\end{equation}
for all $g\in\tLo(W):=\{g\in\tLo:gW=W\}$ and $t\in\Rl$.
\end{remark}

\section{Deformations of quantum field theories on de Sitter spacetime} 
\label{sec:DeformationsQFT}
\subsection{Field nets} 
\label{subsec:FieldNetsOnDS}
We work in the operator-algebraic approach to quantum field theory on curved spacetimes \cite{Dimock80} adapted to the concrete case of de Sitter space \cite{BorchersBuchholz99}.
To this end, we consider a C$^*$-algebra $\F$ (field algebra) whose elements are physically interpreted as (bounded functions of) quantum fields on $M$. We equip $\F$ with a local structure and focus on localization in wedges, since this turns out to be stable under the deformation. Hence we associate to each $W\in\W$ a C$^*$-subalgebra $\F(W)\subset \F$. Due to the trivial inclusion properties of wedges in de Sitter space (see Lemma \ref{lem:WedgeInclusions}) the usual isotony condition reduces to well-definedness of $W\mapsto\F(W)$.

We assume that there exists a strongly continuous representation $\alpha$ of $\tLo$ by automorphisms on $\F$, such that
\begin{itemize}
 \item[1)] (De Sitter Covariance): for all $g\in\tLo$, $W\in\W$ holds 
\begin{equation*}
\label{eq:DeSitterCovariance}
	\alpha_g(\F(W))=\F(gW). 
\end{equation*}
\end{itemize}

\noindent
Furthermore, we assume that there is a Lie group $G$ (global gauge group) and a strongly continuous representation $\sigma$ of $G$ by automorphisms on $\F$, such that
\begin{itemize}
 \item[2)] (Gauge Invariance): for all $h\in G$, $g\in\tLo$, $W\in\W$ holds 
\begin{equation}
\label{eq:GaugeInvariance}
	\sigma_h(\F(W))=\F(W),\qquad \sigma_h\circ\alpha_g=\alpha_g\circ \sigma_h. 
\end{equation}
\end{itemize}

\noindent
We assume that there exists a distinguished element $h_0\in G$ such that $\gamma:=\sigma(h_0)$ satisfies
\begin{equation}
\label{eq:BFauto} 
	\gamma^2=\id.
\end{equation}
This (grading) automorphism can be used to separate an operator $F\in \F(W)$ into its Bose($+$) and Fermi($-$) part via $F_\pm:=(F\pm\gamma(F))/2$. 

\begin{remark}
	For convenience, we assume that the datum $(\{\F(W):W\in\W\},\alpha,\sigma,\gamma)$ is faithfully and covariantly represented on a Hilbert space $\HS$. So to each $\F(W)$ corresponds a norm-closed $*$-subalgebra of $\BH$ and the automorphisms $\alpha,\sigma,\gamma$ are implemented by the adjoint action of unitary operators $U,V,Y$ on $\HS$, respectively. Note that this is no loss of generality since we can either use the covariant representation which exists for C$^*$-dynamical systems (see \cite{BuchholzLechnerSummers10, DappiaggiLechnerMorfa-Morales10} and references therein) or we work in the GNS-representation of a de Sitter- and gauge-invariant state. In the former case we assume that $\HS$ is separable, as it is the case in a variety of concrete models.
\end{remark}

We assume that the grading satisfies $Y^2=1$. With the operator $Y$ a unitary twisting map $Z$ is defined to treat the (anti)commutation relations between the Bose/Fermi parts of a field on the same footing \cite{DoplicherHaagRoberts65}. Let $Z:=(1-iY)/\sqrt{2}$ and 
\begin{equation*}
 \F(W)^t:=Z\F(W)Z^{-1}.
\end{equation*}
The map $F\mapsto ZFZ^{-1}$ is an isomorphism of $\F(W)$ and we have \cite{Foit83}
\begin{equation*}
 \F(W)^{tt}=\F(W),\quad \F(W)^t{}'=\F(W)'{}^t,\quad W\in\W,
\end{equation*}
where the commutant is understood as the relative commutant in $\F$. Locality is now formulated in the following way
\begin{itemize}
 \item[3)] (Twisted Locality): $\F(W)\subset \F(W')^t{}',\;W\in\W$.
\end{itemize}
Twisted locality is equivalent to the ordinary (anti)commutation relations between the Bose/Fermi parts $F_\pm$ of fields, {\it i.e.}, $[F_+,G_\pm]=[F_\pm,G_+]=\{F_-,G_-\}=0$ for $F\in\F(W)$, $G\in\F(W')$, $W\in\W$ (see \cite{DoplicherHaagRoberts65}).

For later reference we define the joint action $\tau:\tLo\times G\ra \Aut(\F)$ of the external and internal symmetry group on $\F$ by
\begin{equation}
\label{eq:jointautomorphism}
 \tau_{g,h}:=\alpha_g\circ \sigma_h,\quad g\in\tLo,\; h\in G.
\end{equation}
The unitary which implements this action is $\U(g,h):=U(g)V(h)$.

\begin{remark}
A datum $(\{\F(W):W\in\W\},\alpha,\sigma,\gamma)$ which satisfies conditions 1) $-$ 3) is referred to as a \emph{wedge-local field net}. We simply write $\F$ to denote it, if no confusion can arise. Examples are nets of $\CAR$-algebras with gauge symmetry, such as the free charged Dirac field.
\end{remark}

\begin{remark}
 Given a field net, the net of observables is defined as
\begin{equation*}
	\A(W):=\{F\in\F(W):\sigma_h(F)=F,\; h\in G\},
\end{equation*}
so observables form the gauge-invariant part of the field net. 
\end{remark}

Due to the transitive action of $\tLo$ on $\W$, it is possible to define a wedge-local field net in terms of an inclusion of just two C$^*$-algebras which are in a suitable relative position to $W_0$. This point of view will be advantageous for the warped convolution later on, since the deformation of a wedge-local field net amounts to deforming the relative position of one algebra in the other. 
Following \cite{BuchholzLechnerSummers10} we make the following definition.
\begin{definition}
\label{def:CausalBorchersSystem}
 A {\em causal Borchers system} $(\F_0,\F,\alpha,\sigma,\gamma)$ relative to $W_0\in\W$ consists of
 \begin{itemize}
  \item[$-$] an inclusion $\F_0\subset \F$ of concrete C$^*$-algebras,
  \item[$-$] commuting representations $\alpha:\tLo\ra \Aut(\F)$ and $\sigma:G\ra \Aut(\F)$ which are unitarily implemented,
  \item[$-$] an automorphism $\gamma$ on $\F$ which commutes with $\alpha$ and $\sigma$ and satisfies $\gamma^2=\id$,
 \end{itemize}
 such that
 \begin{enumerate}
  \item[a)] $\alpha_{g}(\F_0)= \F_0,\; g\in\tLo(W_0)$,
  \item[b)] $\alpha_{\bj_{W_0}}(\F_0)\subset (\F_0)^t{}'$,
  \item[c)] $\sigma_h(\F_0)=\F_0$, $h\in G$.
 \end{enumerate}
\end{definition}

\begin{proposition}
\label{prop:fieldquadruple-fieldnet} 
	Let $(\F_0,\F,\alpha,\sigma,\gamma)$ be a causal Borchers system relative to $W_0$. Then 
	\begin{equation}
		\label{eq:WT-WLN} W:=g W_0\longmapsto \alpha_{g}(\F_0)=:\F(W), 
	\end{equation}
	defines a wedge-local field net together with $(\alpha,\sigma,\gamma)$. 
\end{proposition}
\begin{proof}
	We begin by proving well-definedness. From $g_1 W_0=g_2W_0$ follows $g_2^{-1}g_1 W_0= W_0$ and $\alpha_{g_2^{-1}g_1}(\F_0)= \F_0$ by assumption a). Hence $\alpha_{g_1}(\F_0)= \alpha_{g_2}(\F_0)$ and the assertion follows.
	
	Covariance holds by definition. 

	Twisted locality is proved in a similar way. Let $W'=gW_0$. Since $W=g\bj_{W_0}W_0$ there holds
	\begin{equation*}
	    \F(W)=\alpha_{g\bj_{W_0}}(\F_0)
	    \subset\alpha_{g}((\F_0)^t{}')
	    =\alpha_{g}(\F_0)^t{}'=\F(W')^t{}',\quad W\in\W,
	\end{equation*}
	where we used condition b), together with the assumption that each $\alpha_g$, $g\in\tLo$ is a homomorphism which commutes with $\gamma$.
	
	The gauge invariance of the local algebras follows immediately:
	\begin{equation*}	
		\sigma_h(\F(W))=\sigma_h(\alpha_g(\F_0))=\alpha_g(\sigma_h(\F_0))=\alpha_g(\F_0)
	=\F(W),
	\end{equation*}
	since the representations $\alpha,\sigma$ commute and by assumption c). 
\end{proof}

\noindent
Note that the converse of this proposition is trivially true. Given a wedge-local field net, then $\F(W_0)\subset \F$ satisfies property a) by covariance and b) by twisted locality. Property c) holds by definition.

\begin{remark}
A causal Borchers system $(\F_0,\F,\alpha,\sigma,\gamma)$ is closely connected to the notion of a causal Borchers triple \cite{BuchholzLechnerSummers10} on Minkowski spacetime (see also \cite{Lechner10} for the related notion of a wedge triple). In this setting, $\F_0\subset\BH$ is a von Neumann algebra and $\alpha$ is the adjoint action of a unitary representation $U$ of the Poincar\'e group. In addition one assumes that the joint spectrum of the generators of the translations $U\uhr \Rl^4$ is contained in the closed forward lightcone (spectrum condition) and that $\F_0$ admits a cyclic and separating vector (existence of a vacuum state). Gauge transformations are absent in this setting since nets of observables are considered.
\end{remark}

\begin{remark}
 For the sake of brevity we will write $\F_0\subset \F$ to denote a causal Borchers system relative to $W_0$.
\end{remark}

\subsection{Deformations of field nets with $\Uone$ gauge symmetry}
 \label{subsec:DeformationsFieldNets}
Now we apply the warped convolution deformation method to our present setting. Let $\F_0\subset\F$ be a causal Borchers system relative to $W_0$. The basic idea is to define a deformation $(\F_0)_{\xi,\kappa}$ of the small algebra $\F_0$ using a suitable $\Rl^2$-action (see below) in such a way that $(\F_0)_{\xi,\kappa}\subset\F$ is again a causal Borchers system. Then the inclusion $(\F_0)_{\xi,\kappa}\subset \F$ gives rise to another wedge-local field net by Proposition \ref{prop:fieldquadruple-fieldnet}. 
\\

\noindent
For the warped convolution we make the further assumption that the gauge group is $G=\Uone\cong\Rl/2\pi\Zl$. The representation $\sigma$ of $\Uone$ yields a $2\pi$-periodic $\Rl$-action $F\mapsto\sigma_{s}(F)$ by automorphisms on $\F$. The warped convolution is now defined with the $\Rl^2$-action $\tau^\xi$ coming from the one-parameter group of boosts $\widetilde{\Gamma}_{W_0}\subset \tLo$ and the internal symmetry group:
\begin{equation*}
\Rl^2\ni (t,s)\longmapsto 
\tau_{\lambda_{W_0}(t),s}=:\tau^\xi_{t,s}:\F\lra \F.
\end{equation*}
Note that $\widetilde{\Gamma}_{W_0}$ implicitly depends on the Killing field $\xi:=\xi_{W_0}$ which is associated with $W_0$ (see Section \ref{subsec:wedges}). We will use the notation
\begin{equation*}
 \lambda_\xi(t):=\lambda_{W_0}(t),\quad U_\xi(t):=U(\lambda_\xi(t)),\quad \U_\xi(t,s):=\U(\lambda_\xi(t),s).
\end{equation*}

Since the warped convolution is defined is terms of oscillatory integrals of operator-valued functions, we first need to specify suitable smooth elements of the C$^*$-algebra $\F$ for which these integrals are well-defined. The joint action (\ref{eq:jointautomorphism}) is a strongly continuous action of the Lie group $\tLo\times \Uone$ which acts automorphically, and therefore isometrically, on $\F$. The algebra $\F_0$ is, in general, only invariant under the action of the subgroup $\widetilde{\L}_0(W_0)\times \Uone$. Adapted to the present setting, and following \cite{DappiaggiLechnerMorfa-Morales10}, we consider the following notion of smoothness with respect to the subgroup $\widetilde{\Gamma}_{W_0}\times\Uone$. 

\begin{definition}
	An operator $F\in\F$ is called {\it $\xi$-smooth}, if $\Rl^2\ni v\mapsto \tau^\xi_v(F)\in\F$ is smooth in the norm topology of $\F$. The set of all $\xi$-smooth operators in $\F$ is denoted by $\F^\infty_\xi$.
\end{definition}

Note that the set $\F^\infty_\xi$ is a norm-dense $*$-subalgebra of $\F$ (see \cite{Taylor86}). Another ingredient for the definition of the warped convolution is the antisymmetric (real) matrix 
\begin{equation*}
	\theta:= 
	\begin{pmatrix}
		0 & 1\\
		-1 & 0 
	\end{pmatrix}
\end{equation*}
and an arbitrary but fixed real number $\kappa$ which plays the role of a deformation parameter. 

\begin{definition}
The {\em warped convolution} of an operator $F\in \F^\infty_\xi$ is defined as 
\begin{equation}
	\label{eq:WarpedConvolution} 
	F_{\xi,\kappa}:=\frac{1}{4\pi^2}\lim_{\eps\ra 0}\int_{\Rl^2\times \Rl^2}dv\, d v'\; e^{-ivv'}\, \chi(\eps v,\eps v')\, \tau^\xi_{\kappa\theta v}(F)\U_\xi(v').
\end{equation}
Here $vv'$ denotes the standard Euclidean inner product of $v,v'\in\Rl^2$ and $\chi\in C^\infty_0(\Rl^2\times \Rl^2)$, $\chi(0,0)=1$ is a cutoff function which is necessary to define this operator-valued integral in an oscillatory sense.  
\end{definition}

From the results in \cite{BuchholzLechnerSummers10} follows that the above limit exists in the strong operator topology of $\BH$ on the dense domain 
\begin{equation*}
	\HS^\infty:=\{\Phi\in\HS: \tLo\times\Uone\ni (g,s) \mapsto \U(g,s)\Phi\in \HS \text{ is smooth in } \|\cdot \|_{\HS}\} 
\end{equation*}
and is independent of the chosen cutoff function $\chi$ within the specified class. The densely defined operator $F_{\xi,\kappa}$ extends to a bounded and smooth operator, which is denoted by the same symbol. 

\begin{definition}
    The space of all vectors which are smooth with respect to the representation $\U_\xi$ is denoted by $\HS^\infty_\xi$ .
\end{definition}

Furthermore, it is shown in \cite{BuchholzLechnerSummers10} that the warped convolution (\ref{eq:WarpedConvolution}) is closely related to Rieffel deformations of C$^*$-algebras \cite{Rieffel93}. In this context one defines, instead of a deformation of the algebra elements, a new product $\times_{\xi,\kappa}$ on $\F^\infty_\xi$ by
\begin{equation*}
 F\times_{\xi,\kappa}F':=\frac{1}{4\pi^2}\lim_{\eps\ra 0}\int_{\Rl^2\times\Rl^2}d v\,d v' e^{-ivv'}
\chi(\eps v,\eps v') \tau^\xi_{\kappa \theta v}(F)\tau^\xi_{v'}(F').
\end{equation*}
This limit exists in the norm-topology of $\F$ for all $F,F'\in\F^\infty_\xi$ and $F\times_{\xi,\kappa}F'$ is again in $\F^\infty_\xi$. The completion of $(\F^\infty_\xi,\times_{\xi,\kappa})$ in a suitable norm yields another C$^*$-algebra \cite{Rieffel93}.

The following lemma collects the basic properties of the map $F\mapsto F_{\xi,\kappa}$ and shows that the warped operators form a representation of the Rieffel deformed C$^*$-algebra for a fixed deformation parameter.

\begin{lemma}[\cite{BuchholzLechnerSummers10,DappiaggiLechnerMorfa-Morales10}]
\label{lem:WCbasicproperties}
 Let $F,F'\in\F^\infty_\xi$ and $\kappa\in\Rl$. Then
\begin{enumerate}
 \item[a)] $(F_{\xi,\kappa})^*=(F^*)_{\xi,\kappa}$.
 \item[b)] $F_{\xi,\kappa}F'_{\xi,\kappa}=(F\times_{\xi,\kappa}F')_{\xi,\kappa}$.
 \item[c)] If $[\tau^\xi_v(F),F']=0$ for all $v\in\Rl^2$, then $[F_{\xi,\kappa},F'_{\xi,-\kappa}]=0$.
 \item[d)] If $[Z\tau^\xi_v(F)Z^*,F']=0$ for all $v\in\Rl^2$, then $[ZF_{\xi,\kappa}Z^*,F'_{\xi,-\kappa}]=0$. 
 \item[e)] Let $X\in\BH$ be a unitary which commutes with $\U_\xi(v)$ for all $v\in\Rl^2$. Then 
$XF_{\xi,\kappa}X^{-1}=(XFX^{-1})_{\xi,\kappa}$ and $XF_{\xi,\kappa}X^{-1}$ is $\xi$-smooth.
 
\end{enumerate}
\end{lemma}
\begin{proof}
 Statements a), b), c), e) were shown in \cite{BuchholzLechnerSummers10}. d) is a consequence of Lemma 3.2 in \cite{DappiaggiLechnerMorfa-Morales10} and the fact that $\gamma$ commutes with $\alpha$ and $\sigma$.
\end{proof}

The next lemma lists the transformation properties of warped operators under the de Sitter and gauge group.

\begin{lemma}
\label{lem:WCcovariance}
 Let $F\in\F^\infty_\xi$, $\kappa\in\Rl$, $g\in\tLo$ and $s\in\Uone$. Then
\begin{enumerate}
 \item[a)] $\alpha_g(F)$ is $g_*\xi$-smooth and 
\begin{equation}
\label{eq:DeformedOperatorsDeSitterCovariance}
 \alpha_g(F_{\xi,\kappa})=\alpha_g(F)_{g_*\xi,\kappa},
\end{equation}
where $g_*\xi$ is the push-forward of $\xi$ with respect to $g$.
\item[b)] $\sigma_{s}(F)$ is $\xi$-smooth and 
\begin{equation}
\label{eq:DeformedOperatorsGaugeCovariance}
 \sigma_{s}(F_{\xi,\kappa})=\sigma_{s}(F)_{\xi,\kappa}.
\end{equation}
\end{enumerate}
\end{lemma}
\begin{proof}
 Statement a) follows from Lemma 3.3 a) in \cite{DappiaggiLechnerMorfa-Morales10} and the fact that $\alpha$ and $\sigma$ commute. Statement b) follows from Lemma \ref{lem:WCbasicproperties} e). 
\end{proof}

Now we apply the warped convolution deformation method to a causal Borchers system $\F_0\subset\F$. Define
\begin{equation*}
 (\F_0)_{\xi,\kappa}:=\overline{\{F_{\xi,\kappa}:F\in\F_0\cap\F^\infty_\xi\}}^{\|\cdot\|}.
\end{equation*}
The following theorem shows that the inclusion $(\F_0)_{\xi,\kappa}\subset\F$ gives rise to a wedge-local field net in the sense of Proposition \ref{prop:fieldquadruple-fieldnet}.

\begin{theorem}
\label{thm:main}
 Let $(\F_0)_{\xi,\kappa}$ be as above. Then
\begin{enumerate}
 		\item[a)] $\alpha_g((\F_0)_{\xi,\kappa})= (\F_0)_{\xi,\kappa},\; g\in\tLo(W_0)$, 
		\item[b)] $\alpha_{\bj_{W_0}}((\F_0)_{\xi,\kappa})\subset ((\F_0)_{\xi,\kappa})^t{}'$,
		\item[c)] $\sigma_{s}((\F_0)_{\xi,\kappa})=(\F_0)_{\xi,\kappa},\; s\in \Uone$.
\end{enumerate}
\end{theorem}
\begin{proof}
	a): Let $F\in\F_0\cap\F^\infty_\xi$ and $g\in\tLo(W_0)$. From (\ref{eq:CovarianceGammaJ}) follows that $g$ commutes with each $\lambda_\xi(t),t\in\Rl$. Hence 
	\begin{equation*}
	    \alpha_g(F_{\xi,\kappa})=\alpha_g(F)_{g_*\xi,\kappa}=\alpha_g(F)_{\xi,\kappa}
	\end{equation*}
	by Lemma \ref{lem:WCcovariance} a) and $\alpha_g(F)\in \F_0$ by property a) of the undeformed causal Borchers system. Therefore $\alpha_g(F_{\xi,\kappa})\in(\F_0)_{\xi,\kappa}$ and by taking the norm-closure of $\{F_{\xi,\kappa}:F\in\F_0\cap\F^\infty_\xi\}$ the statement $\alpha_g((\F_0)_{\xi,\kappa})=(\F_0)_{\xi,\kappa}$ follows.

	b): From Lemma \ref{lem:WCcovariance} a) and (\ref{eq:CovarianceGammaJ}) follows
	\begin{equation}
	\label{eq:proof_locality}
	    \alpha_{{\bj}_{W_0}}(F_{\xi,\kappa})=\alpha_{{\bj}_{W_0}}(F)_{{{\bj}_{W_0}}_*\xi,\kappa}=\alpha_{{\bj}_{W_0}}(F)_{\xi,-\kappa},
	\end{equation}
	together with an elementary substitution in (\ref{eq:WarpedConvolution}). We have $\alpha_{{\bj}_{W_0}}(F)\in(\F_0)^t{}'$ by property b) of the undeformed causal Borchers system, {\it i.e.}, $[Z\alpha_{{\bj}_{W_0}}(F)Z^{-1},F']=0$ for all $F'\in\F_0$. Pick some $F'\in\F_0\cap\F^\infty_\xi$ and consider its warped convolution $F'_{\xi,\kappa}$. We have $[Z\tau_v^\xi(\alpha_{{\bj}_{W_0}}(F))Z^{-1},F']=0$ for all $v\in\Rl^2$ since $\F_0$ is invariant under $\widetilde{\Gamma}_{W_0}\times\Uone$. Hence 
	\begin{equation*}
	  [Z\alpha_{\bj_{W_0}}(F_{\xi,\kappa})Z^{-1},F'_{\xi,\kappa}]=
	  \alpha_{\bj_{W_0}}([ZF_{\xi,\kappa}Z^{-1},\alpha_{\bj_{W_0}}(F'_{\xi,\kappa})])=
	  \alpha_{\bj_{W_0}}([ZF_{\xi,\kappa}Z^{-1},F'_{\xi,-\kappa}])=0
	\end{equation*}
	by (\ref{eq:proof_locality}) and Lemma \ref{lem:WCbasicproperties} e). By taking the norm-closure of $\{F_{\xi,\kappa}:F\in\F_0\cap\F^\infty_\xi\}$ the statement $\alpha_{\bj_{W_0}}((\F_0)_{\xi,\kappa})\subset ((\F_0)_{\xi,\kappa})^t{}'$ follows.

	Assertion c) is a consequence of Lemma \ref{lem:WCcovariance} b) and the invariance $\F_0$ under gauge transformations.
\end{proof}

\begin{remark}
 Note that the minus sign which appears in (\ref{eq:proof_locality}) is the main reason why the locality proof works. That this argument is also valid for the extended symmetry group $\L_+$ can be seen in the following way. The reflection $\widehat{j}_{W_0}(x^0,x^1,\vec{x})=(-x^0,-x^1,\vec{x})$ commutes with boosts in the $x^1$-direction. Again, a deformed operator transforms under the lift $\widehat{\bj}_{W_0}$ of $\hat{j}_{W_0}$ according to $\alpha_{\widehat{\bj}_{W_0}}(F_{\xi,\kappa})=F_{\xi,-\kappa}$ since $\widehat{\bj}_{W_0}$ is represented by an {\em anti}unitary operator.
\end{remark}

\subsection{Example: Deformations of $\CAR$-nets} 
\label{subsec:CARnets}
Now we investigate a particular class of wedge-local field nets in more detail, namely, nets of $\CAR$-algebras. The free charged Dirac field is an example thereof. After it is shown that these models fit into the framework of Section \ref{subsec:DeformationsFieldNets}, the properties of the deformed field operators and observables are studied in detail and it is proved that the deformed and undeformed nets are non-isomorphic.

\subsubsection{The selfdual $\CAR$-algebra}
We use Araki's selfdual approach to the $\CAR$-algebra \cite{Araki71}. Let $H$ be a separable infinite-dimensional complex Hilbert space with inner product $\left<.\,,.\right>$ and let $C$ be an antiunitary involution on $H$, {\it i.e.}, $C^2=1$ and $\left<C f_1,C f_2\right>=\left<f_2,f_1\right>$ for all $f_1,f_2\in H$. 
On the $*$-algebra $\CAR_0(H,C)$ which is algebraically generated by elements $B(f),f\in H$ and a unit $1$, satisfying
\begin{itemize}
 \item[$a)$] $f\mapsto B(f)$ is complex linear,
 \item[$b)$] $B(f)^*=B(C f)$,
 \item[$c)$] $\{B(f_1),B(f_2)\}=\left<C f_1,f_2\right>\cdot 1$,
\end{itemize}
there exists a unique C$^*$-norm satisfying (see \cite{EvansKawahigashi98})
\begin{equation*}
	\|B(f)\|^2=\frac{1}{2}(\|f\|^2+\sqrt{\|f\|^4-|\left<f,Cf\right>|^2}).
\end{equation*}
Hence each $B(f)$ is bounded and $f\mapsto B(f)$ is norm-continuous. The C$^*$-completion of $\CAR_0(H,C)$ is denoted by $\CAR(H,C)$. This C$^*$-algebra is simple \cite{Araki71}, so all its representations are faithful or trivial. 

If $u$ is a unitary on $H$ which commutes $C$, then $\alpha_u(B(f)):=B(uf)$ defines a $*$-automorphism on $\CAR(H,C)$. We refer to $u$ as Bogolyubov transformation and to $\alpha_u$ as Bogolyubov automorphism.
\subsubsection{Quasifree representations}
A state $\omega$ on $\CAR(H,C)$ is called quasifree, if
\begin{align*}
	\omega(B(f_1)\cdots B(f_{2n+1}))&=0\\
	\omega(B(f_1)\cdots B(f_{2n}))
	&=(-1)^{n(n-1)/2}\sum_{\epsilon}\mathrm{sgn}(\epsilon)\prod_{j=1}^n\omega(B(f_{\epsilon(j)})B(f_{\epsilon(j+n)}))
\end{align*}
holds for all $n\in\Nl$, where the sum runs over all permutations $\epsilon$ of $\{1,\dots,2n\}$ satisfying
\begin{equation*}
	\epsilon(1)<\dots <\epsilon(n),\quad \epsilon(j)<\epsilon(j+n),\quad j=1,\dots,n.
\end{equation*}
Let $S$ be a bounded linear operator on $H$ satisfying
\begin{equation*}
	\label{eq:CAR_Soperator}
	S=S^*,\quad 0\le S\le 1,\quad CSC=1-S.
\end{equation*}
In \cite[Lemma 3.3]{Araki71} it is shown that for every such $S$ there exists a unique quasifree state $\omega_S$ satisfying 
\begin{equation*}
\label{eq:QuasifreeState}
	\omega_S(B(f)B(g))=\left<Cf,Sg\right>.
\end{equation*}
Conversely, every quasifree state on $\CAR(H,C)$ gives rise to such an operator \cite[Lemma 3.2]{Araki71}. Hence quasifree states can be parametrized by this class of operators. 

Let $\omega_S$ be a quasifree state. For the GNS-triple associated with $(\CAR(H,C),\omega_S)$ we write $(\HS_S,\pi_S,\Omega_S)$. If a Bogolyubov transformation $u$ commutes with $S$, then the associated Bogolyubov automorphism can be unitarily implemented, {\em i.e.}, there exists a unitary operator $U_S$ on $\HS_S$, such that
\begin{equation*}
	\pi_S(\alpha_u(F))=U_S\pi_S(F)U_S^{-1},\quad U_S\Omega_S=\Omega_S
\end{equation*}
holds for all $F\in\CAR(H,C)$ (see \cite[Lemma 4.2]{Araki71}).

Fock states are a particular class of quasifree states where $S=P$ is a projection. The GNS Hilbert space $\HS_P$ is the Fermionic Fock space over $PH$
\begin{equation*}
	\HS_P=\Cl\oplus \bigoplus_{n\ge 1}\wedge^n PH,
\end{equation*}
where $\wedge^n PH$ denotes the antisymmetrization of the $n$-fold tensor product of $PH$, the cyclic vector $\Omega_P$ is the Fock vacuum in $\HS_P$ and
\begin{equation*}
 \pi_P(B(f))=a^*(PC f)+a(Pf),
\end{equation*}
with the standard Fermi creation and annihilation operators $a^\#(Pf)$ on $\HS_P$.
Two representations $(\HS_P,\pi_P)$ and $(\HS_{P'},\pi_{P'})$ are unitarily equivalent if and only if $P-P'$ is Hilbert-Schmidt (see \cite{Araki71} and \cite{ShaleStinespring64}). As a consequence, a Bogolyubov transformation $u$ is unitarily implementable if and only if $[u,P]$ is Hilbert-Schmidt.

\subsubsection{Nets of $\CAR$-algebras}
In order to introduce charges and global gauge transformations we double the Hilbert space $\hs:=H\oplus H$ and define the antiunitary involution
\begin{equation*}
 \cc:=
\begin{pmatrix}
 0 & C\\
 C & 0
\end{pmatrix}.
\end{equation*}
For the inner product on $\hs$ we write $(.\,,.)$. Applying Araki's construction to $(\hs,\cc)$ yields again a C$^*$-algebra $\CAR(\hs,\cc)$. The unitary operators
\begin{equation}
\label{eq:CAR_GaugeTransformationGeneral}
	v(s)(f_+\oplus f_-):=e^{is}f_+\oplus e^{-is}f_-,\quad s\in\Rl,\;f_\pm\in H
\end{equation}
commute with $\cc$ so there exists a representation $\sigma:\Uone\ra \Aut(\CAR(\hs,\cc))$, such that
\begin{equation}
\label{eq:CAR_gaugetransformation}
 \sigma_{s}(B(f))=B(v(s)f).
\end{equation}
We assume that there exists a unitary representation $u$ of $\tLo$ on $\hs$ which commutes with $\cc$ so that there exists a representation $\alpha:\tLo\ra\Aut(\CAR(\hs,\cc))$ satisfying
\begin{equation*}
\label{eq:CAR_covariance}
 \alpha_g(B(f))=B(u(g)f).
\end{equation*}
For the representers of the subgroup $\widetilde{\Gamma}_{W_0}$ we write $u_\xi(t):=u(\lambda_\xi(t)),\, t\in\Rl$.
\begin{remark}
 The picture in terms of spinors and cospinors is obtained by setting 
\begin{equation}
\label{eq:CAR_cospinors}
 \Psi(f_-):=B(0\oplus f_-),\quad \Psi^\dagger(f_+):=B(f_+\oplus 0).
\end{equation}
There holds $\Psi(f_-)^*=\Psi^\dagger(Cf_-), \Psi^\dagger(f_+)^*=\Psi(Cf_+)$ and from the linearity of $f\mapsto B(f)$ follows that (co)spinors transform according to
\begin{equation}
\label{eq:CAR_cospinors_gauge}
 \sigma_s(\Psi(f_-))=e^{-is}\Psi(f_-),\quad \sigma_s(\Psi^\dagger(f_+))=e^{is}\Psi^\dagger(f_+)
\end{equation}
under gauge transformations.
\end{remark}

Now we come to the net structure of the theory. Let $\hs_0\subset \hs$ be a complex linear subspace satisfying $\cc \hs_0\subset \hs_0$ and
\begin{itemize}
 \item[$i)$] $u(g)\hs_0=\hs_0,\;g\in\tLo(W_0)$,
 \item[$ii)$] $u(\bj_{W_0})\hs_0\subset (\hs_0)^\perp$,
 \item[$iii)$] $v(s)\hs_0=\hs_0,\; s\in\Rl$,
\end{itemize}
where $(\hs_0)^\perp$ is the orthogonal complement of $\hs_0$. 
\begin{remark}
In concrete models this space is explicitly given and can be constructed by different methods. In the case of the free charged Dirac field the space $\hs_0$ can be defined as the set of (Fourier-Helgason transforms of) spinor-valued testfunctions on $M$ which are localized in the wedge $W_0$ (see \cite{BartesaghiGazeauMoschellaTakook} and \cite{BrosMoschella95} for the scalar free field case) or one considers smooth sections of the Dirac bundle over $M$ modulo the kernel of the causal propagator which is associated with the Dirac equation \cite{Dimock82}, \cite{Sanders2010}. Since this space is constructed from testfunctions it is clear that conditions $i),ii)$ and $iii)$ are satisfied.  
\end{remark}

It is an easy exercise to show that the above conditions imply that 
\begin{equation*}
 W:=gW_0\mapsto u(g)\hs_0=:\hs(W)
\end{equation*}
is an isotonous, $\tLo$-covariant, wedge-local and gauge-invariant net of complex Hilbert spaces in the sense of \cite{BaumgaertelJurkeLledo94}. Hence it is an immediate consequence of Araki's construction that 
\begin{equation}
\label{eq:CARnet}
	W\mapsto \CAR(\hs(W),\cc)=:\F(W)\subset\F:=\CAR(\hs,\cc)
\end{equation}
is a wedge-local field net. Equivalently, from conditions $i)$, $ii)$ and $iii)$ follows that the inclusion $\CAR(\hs_0,\cc)=:\F_0\subset\F=\CAR(\hs,\cc)$ satisfies
\begin{itemize}
 \item[$-$] $\alpha_g(\F_0)=\F_0,\; g\in\tLo(W_0)$
 \item[$-$] $\alpha_{\bj_{W_0}}(\F_0)\subset (\F_0)^t{}'$
 \item[$-$] $\sigma_s(\F_0)=\F_0,\; s\in\Rl$
\end{itemize}
and $gW_0\mapsto \alpha_g(\F_0)$ defines a wedge-local field net by Proposition \ref{prop:fieldquadruple-fieldnet}) which coincides with (\ref{eq:CARnet}). 

\begin{remark}
    Observables in this net  are polynomials of $\Psi(f_-)\Psi^\dagger(f_+)$ which are manifestly gauge-invariant. The quasilocal algebra, generated by them, is denoted by $\A$. 
\end{remark}

From the above discussion it is clear that $W\mapsto \F(W)$ complies with the general assumptions of Section \ref{subsec:FieldNetsOnDS}. As we mentioned before, the algebra $\F$ contains a norm-dense $*$-subalgebra of smooth elements $\F^\infty_\xi$. These can be constructed by smoothening out any element $F\in \F$ with a smooth and compactly supported function $f\in C^\infty_0(\widetilde{\Gamma}_{W_0}\times \Uone)$ via
\begin{equation*}
 F_f:=\int_{\widetilde{\Gamma}_{W_0}\times \Uone}d(g,h)\tau^\xi_{g,h}(F)f(g,h)
\end{equation*}
where $d(g,h)$ is the left-invariant Haar measure on $\widetilde{\Gamma}_{W_0}\times \Uone$. By choosing sequences of functions $f_n$ which converges to the Dirac delta measure at the identity of $\widetilde{\Gamma}_{W_0}\times \Uone$ one sees that these elements are dense in $\F$ in the norm topology. Since the subalgebra $\F_0$ is invariant under the action $\tau^\xi$ it also contains a norm-dense $*$-subalgebra of smooth elements.
For the warped convolution $gW_0\mapsto \alpha_g((\F_0)_{\xi,\kappa})$ of a net of $\CAR$-algebras we will use the shorthand notation $\F_\kappa$.

\begin{remark}
	As we mentioned before, the free charged Dirac field provides an explicit example of a wedge-local field net of $\CAR$-algebras. For spin 1/2 fields there exists a unique de Sitter-invariant state with the Hadamard property \cite{AllenJacobson86}, \cite{AllenLuetken86}. It is the analogue of the Bunch-Davies state \cite{Allen85} in the the spin 1/2 case. The Dirac field in this representation was studied in \cite{BartesaghiGazeauMoschellaTakook} and it was shown that it satisfies the so-called ``geometric KMS-condition''. By the same methods as in \cite{BorchersBuchholz99} one can prove that this condition implies the Reeh-Schlieder property of the state.
\end{remark}

\subsubsection{Deformation fix-points for observables}
Let $\F$ be a $\CAR$-net over $(\hs,\cc)$ in a quasifree representation $(\HS_S,\pi_S,\Omega_S)$ of a de Sitter- and gauge-invariant state. 
\begin{remark}
	As we mentioned before, all representations of the $\CAR$-algebra are faithful so will omit the $S$-dependence in our notation from now on. 
\end{remark}

\noindent
For the implementing operators we write
\begin{equation*}
	\pi(\alpha_g(B(f)))=U(g)\pi(B(f))U(g)^{-1},\qquad
	\pi(\sigma_s(B(f)))=V(s)\pi(B(f))V(s)^{-1}.
\end{equation*}
As $\alpha$ and $\sigma$ are strongly continuous, the representations $U$ and $V$ are also strongly continuous. Stone's theorem implies that the one-parameter group $\{V(s):s\in\Rl\}$ has a unique self-adjoint generator $Q$ with spectrum $\S\subset\Zl$ since $V(2\pi)=1$. Hence the representation space $\HS$ is $\S$-graded (charged sectors)
\begin{equation}
	\label{eq:QuasifreeDecomposition}
 \HS=\bigoplus_{n\in \cal{S}}\HS_n,\quad \HS_n=\{\Phi\in\HS:Q\Phi=n\Phi\}.
\end{equation}
From the transformation properties (\ref{eq:CAR_cospinors_gauge}) for (co)spinors follows that $\pi(\Psi(f_-))$ decreases charges by one and $\pi(\Psi^\dagger(f_+))$ increases charges by one, {\em i.e.}, 
\begin{equation*}
    \pi(\Psi(f_-))\HS_n\subset \HS_{n-1},\quad 
    \pi(\Psi^\dagger(f_+))\HS_n\subset \HS_{n+1}
\end{equation*}
In the following we will frequently use the spectral decomposition $V(s)=\sum_{n\in \cal{S}}e^{isn}E(n)$, where $E(n)$ is the projector onto the eigenspace $\HS_n$ of $Q$. 

Before we determine the fix-points of the deformation map for observables, we compute the warped convolution for intertwiners between charged sectors.
\begin{proposition}
	\label{prop:CARDeformation}
 Let $\pi(F)\in\BH$ be $\xi$-smooth such that $\pi(F)\HS_n\subset\HS_{n+m}$. Then
\begin{equation*}
 \pi(F)_{\xi,\kappa}=\sum_{n\in \cal{S}}U_\xi(\kappa n)\pi(F)U_\xi(-\kappa(n+m))E(n).
\end{equation*}
\end{proposition}
\begin{proof}
 Let $\Phi\in\HS^\infty_\xi$. Then
\begin{align*}
 \pi(F)_{\xi,\kappa}\Phi
&=\pi(F)_{\xi,\kappa}\sum_{n\in \cal{S}}E(n)\Phi\\
&=\sum_{n\in \cal{S}}\pi(F)_{\xi,\kappa}E(n)\Phi\\
&=\frac{1}{4\pi^2}\sum_{n\in \cal{S}}\lim_{\eps\ra 0}\int dv \int dv'e^{-ivv'}\chi(\eps v,\eps v')\U_\xi(\kappa\theta v)\pi(F)\U_\xi(-\kappa\theta v)^{-1}\U_\xi(v')E(n)\Phi\\
&=\frac{1}{4\pi^2}\sum_{n\in \cal{S}}\lim_{\eps \ra 0}\int dt ds \int dt' ds' e^{-i(tt'+ss')}
\chi_1(\eps t,\eps t')\chi_2(\eps s,\eps s')\cdot\\
&\hspace{4cm}\cdot
U_\xi(\kappa s)V(-\kappa t)\pi(F)V(-\kappa t)^{-1}U_\xi(\kappa s)^{-1}U_\xi(t')V(s')E(n)\Phi\\
&=\frac{1}{4\pi^2}\sum_{n\in \cal{S}}\lim_{\eps_1\ra 0}\int dt dt' \lim_{\eps_2\ra 0}\int ds ds' e^{-i(tt'+ss')}
\chi_1(\eps_1 t,\eps_1t')\chi_2(\eps_2 s,\eps_2s')\cdot\\
&\hspace{4cm}\cdot
U_\xi(\kappa s)V(-\kappa t)\pi(F)V(-\kappa t)^{-1}U_\xi(\kappa s)^{-1}U_\xi(t')V(s')E(n)\Phi\\
&=\frac{1}{4\pi^2}\sum_{n\in \cal{S}}\lim_{\eps_1\ra 0}\int dt dt' \lim_{\eps_2\ra 0}\int ds ds' e^{-i(tt'+ss')}
\chi_1(\eps_1 t,\eps_1t')\chi_2(\eps_2 s,\eps_2s')\cdot\\
&\hspace{4cm}\cdot
U_\xi(\kappa s)e^{-i\kappa tm}\pi(F) U_\xi(\kappa s)^{-1}U_\xi(t')e^{is'n}E(n)\Phi\\
&=\frac{1}{4\pi^2}\sum_{n\in \cal{S}}\lim_{\eps_1\ra 0}\int dt dt' \lim_{\eps_2\ra 0}\int ds ds' e^{-it(t'+\kappa m)}e^{-is'(s-n)}
\chi_1(\eps_1 t,\eps_1t')\chi_2(\eps_2 s,\eps_2s')\cdot\\
&\hspace{4cm}\cdot
U_\xi(\kappa s)\pi(F) U_\xi(\kappa s)^{-1}U_\xi(t')E(n)\Phi\\
&=\frac{1}{2\pi}\sum_{n\in \cal{S}}\lim_{\eps_1\ra 0}\int dt dt' e^{-it(t'+\kappa m)}
\chi_1(\eps_1 t,\eps_1t') U_\xi(\kappa n)\pi(F) U_\xi(\kappa n)^{-1}U_\xi(t')E(n)\Phi\\
&=\sum_{n\in \cal{S}}U_\xi(\kappa n)\pi(F) U_\xi(\kappa n)^{-1}U_\xi(-\kappa m)E(n)\Phi.
\end{align*}
In the first line we used the strong convergence of $\sum_{n\in \cal{S}}E(n)$ to the identity and the continuity of $\pi(F)_{\xi,\kappa}$ as an operator on $\HS$ for the second equality. Since the definition of the warped convolution (\ref{eq:WarpedConvolution}) does not depend on the cut-off function $\chi$ we choose $\chi(t,s,t',s')=\chi_1(t,t')\chi_2(s,s')$ with $\chi_l\in C^\infty_0(\Rl\times \Rl)$, $\chi_l(0,0)=1$, $l=1,2$.
For the fifth equality we use Fubini and regularize the integrals in the variables $s,s'$ and $t,t'$ separately by introducing cutoffs $\eps_1,\eps_2$ (see \cite{Rieffel93}). The behavior of $\pi(F)$ under gauge transformations and $V(s')E(n)=e^{is'n}E(n)$ is used in the sixth line. After that the $s'$-integration is performed and the Fourier transform of $\chi_2$ yields a factor $2\pi\delta(s-n)$ in the limit $\eps_2\ra 0$ since $\chi_2(0,0)=1$. Similarly we obtain a factor $2\pi\delta(t'+\kappa m)$ in the limit $\eps_1\ra 0$.
\end{proof}

\begin{remark}
    Specializing this proposition to $m=0$ yields the warped convolution for observables and $m=\pm1$ for (co)spinors.
\end{remark}

\noindent
\\
Next we determine the fix-points of the map $\pi(A)\mapsto \pi(A)_{\xi,\kappa}$ for observables. For this purpose we need some basic facts about one-parameter unitary groups. The unitary operators $\{U_\xi(t):t\in\Rl\}$ form a strongly continuous one-parameter group and by Stone's theorem there exists a unique selfadjoint and (in general) unbounded operator $K_\xi$ (the generator the group) which is defined as
\begin{equation}
\label{eq:DefinitionGenerator}
	iK_\xi\Phi=\lim_{t\ra 0}\frac{1}{t}\bl(U_\xi(t)\Phi-\Phi)
\end{equation} 
on the dense domain $D(K_\xi)=\{\Phi\in \HS:\lim_{t\ra 0}\bl(U_\xi(t)\Phi-\Phi)/t\text{ exists}\}$. 
For elements in $D(K_\xi)$ where $t\mapsto U_\xi(t)\Phi$ is smooth in $\|\cdot \|_{\HS}$ we write $D(K_\xi)^\infty$. Note that $D(K_\xi)^\infty$ is dense in $\HS$ and $\Phi\in D(K_\xi)^\infty$ if and only if $\Phi\in D((K_\xi)^l)$ for all $l\ge 1$. If an operator commutes with $U_\xi(t)$ for all $t\in\Rl$, then $D(K_\xi)^\infty$ is invariant under its action. In particular we have
\begin{equation}
\label{eq:InvarianceDomain}
	U_\xi(t)D(K_\xi)^\infty\subset D(K_\xi)^\infty,\;t\in\Rl,\qquad E(n)D(K_\xi)^\infty\subset D(K_\xi)^\infty,\; n\in\S.
\end{equation}
Furthermore, for $F\in \F^\infty_\xi$ there holds $\pi(F)D(K_\xi)^\infty\subset D(K_\xi)^\infty$ since $F$ is smooth with respect to boosts.

Since observables $A\in\F$ are gauge-invariant, it follows that they are diagonal with respect to the orthogonal decomposition (\ref{eq:QuasifreeDecomposition}) of the representation space:
\begin{equation*}
	\pi(A)=\bigoplus_{n\in \S}\pi_n(A), \qquad \pi_n(A)=\pi(A)E(n):\HS_n\ra \HS_n.
\end{equation*} 		
Furthermore, since observables commute with gauge unitaries and leave charged sectors invariant, it follows that each $\pi_n:\A\subset\F\ra\cal{B}(\HS_n)$ is a representation of the quasilocal algebra $\A$. As all representations of the $\CAR$-algebra are faithful, each $\pi_n,n\in\S$ is faithful. So if $\pi_n(A)=0$ for some $n\in \S$, then $A=0$, which implies $\pi_{m}(A)=0$ for all $m\in \S$ by linearity.

Obviously, from Proposition \ref{prop:CARDeformation} follows that $\pi(A)$ is invariant under the deformation if $\pi_n(A)=1$ for all $n\neq 0$, since $\sum_{n\in S}E(n)$ converges strongly to the identity.

\begin{proposition}
	Let $\A$ be the net of observables in a $\mathrm{CAR}$-net $\F$ in a quasifree representation of a de Sitter- and gauge-invariant state with Reeh-Schlieder property.
	Let $A\in\A(W_0)$ be a $\xi$-smooth observable. If there exists an $\veps\in\Rl$ such that $\pi(A)_{\xi,\kappa}=\pi(A)$ for all $|\kappa|<\veps$, then $\pi(A)\in\Cl\cdot 1$.
\end{proposition}
\begin{proof}
	From the linearity of $\pi(A)\mapsto \pi(A)_{\xi,\kappa}$ together with the fact that each projection $E(n)$ is linear and commutes with boosts and gauge transformations follows
	\begin{equation}
	\label{eq:DecompositionObservableDeformed}
		\pi(A)_{\xi,\kappa}
		=\bigoplus_{n\in\S}\pi_n(A)_{\xi,\kappa}
		=\bigoplus_{n\in\S}\pi(A)_{\xi,\kappa}E(n).
	\end{equation}
	Consider now compact intervals $\Delta,\Delta'\subset\Rl$ and the spectral projections $\widetilde{E}_\xi(\Delta),\widetilde{E}_\xi(\Delta')$ of the generator $K_\xi$. For $\Phi,\Phi'\in D(K_\xi)^\infty$ define vectors $\Phi_\Delta:=\widetilde{E}_\xi(\Delta)\Phi$ and $\Phi'_{\Delta'}:=\widetilde{E}_\xi(\Delta')\Phi'$. As $\kappa\mapsto \pi(A)_{\xi,\kappa}$ is constant in a neighborhood of $\kappa=0$ there follows from (\ref{eq:DecompositionObservableDeformed})
	\begin{equation*}
	 0=(\Phi'_{\Delta'},\frac{d}{d\kappa}\pi(A)_{\xi,\kappa}E(n)\br|_{\kappa=0}\Phi_\Delta)_{\HS}.
	\end{equation*}
	Proposition (\ref{prop:CARDeformation}) for observables ($m=0$) implies
	\begin{equation*}
	 0
	=(\Phi'_{\Delta'},
	\frac{d}{d\kappa}U_\xi(\kappa n)\pi(A)U_\xi(-\kappa n)E(n)\br|_{\kappa=0}
	\Phi_\Delta)_{\HS}
	=in(\Phi'_{\Delta'},[K_\xi,\pi(A)]E(n)\Phi_\Delta)_{\HS}.
	\end{equation*}
	Note that $E(n)\Phi\in D(K_\xi)^\infty$ and  $\pi(A)\Phi\in D(K_\xi)^\infty$ for all $n\in \S,\Phi\in D(K_\xi)^\infty$ due to the invariance properties (\ref{eq:InvarianceDomain}) of $D(K_\xi)^\infty$. Hence
	\begin{equation*}
	 0=(\Phi'_{\Delta'},[(K_\xi)^l,\pi(A)]E(n)\Phi_\Delta)_{\HS}
	=(\Phi'_{\Delta'},[(K_\xi)^l,\pi_n(A)]\Phi_\Delta)_{\HS}
	\end{equation*}
	for all $n\neq 0$, $l\ge 0$ and since $\Phi_\Delta$, $\Delta\subset\Rl$ compact is an analytic vector for $K_\xi$, there follows
	\begin{equation*}
	 0=(\Phi'_{\Delta'},[U_\xi(t),\pi_n(A)]\Phi_\Delta)_{\HS}
	  =\sum_{l\ge 0}\frac{(it)^l}{l!}(\Phi'_{\Delta'},[(K_\xi)^l,\pi_n(A)]E(n)\Phi_\Delta)_{\HS}
	\end{equation*}
	As the linear span of $\{\Phi_\Delta:\Delta\subset\Rl\text{ compact},\,\Phi\in D(K_\xi)^\infty\}$ is dense in $\HS$ (see \cite[p.8]{Taylor86}), the bounded operator $[U_\xi(t),\pi_n(A)]$ vanishes on $\HS$ for $n\neq 0$. However, as
	\begin{equation*}
		0=U_\xi(t)\pi_n(A)U_\xi(t)^{-1}-\pi_n(A)=\pi_n(\alpha_{\lambda_\xi(t)}(A)-A)
	\end{equation*}
	for all $n\neq 0$ implies $\pi_m(\alpha_{\lambda_\xi(t)}(A)-A)=0$ for all $m\in \S$ it follows that $[U_\xi(t),\pi_n(A)]$ vanishes on $\HS$ for all $n\in\cal{S}$.
	
	 Since the GNS vector $\Omega$ is de Sitter invariant there holds
	\begin{equation*}
	 U_\xi(t)\pi(A)\Omega=\pi(A)U_\xi(t)\Omega=\pi(A)\Omega,
	\end{equation*}
	so the vector $\pi(A)\Omega$ is boost-invariant. In \cite[Lemma 2.2]{BorchersBuchholz99} it is shown that boost-invariant vectors must in fact be invariant under the whole de Sitter group. Hence
	\begin{equation*}
		U(g)\pi(A)U(g)^{-1}\Omega=U(g)\pi(A)\Omega=\pi(A)\Omega,\quad g\in \tLo.
	\end{equation*}
	From the Reeh-Schlieder property of the state follows $U(g)\pi(A)U(g)^{-1}=\pi(A)$ since $\Omega$ is separating for $\A(W_0)$. Pick $g={\bj}_{W_0}$ and we find
	\begin{equation*}
	 \pi(A)\in \A(W_0)''\cap \alpha_{{\bj}_{W_0}}(\A(W_0)'')\subset \A(W_0)''\cap (\A(W_0)'')'.
	\end{equation*}
	by locality. Powers and St\o rmer \cite{PowersStormer} have shown that every quasifree and gauge-invariant representation of a $\CAR$-algebra is primary so the local algebras are factors which implies $\pi(A)\in \Cl\cdot 1$.
\end{proof}

\subsubsection{Unitary inequivalence}
Let $\F$ be a $\CAR$-net over $(\hs,\cc)$ in a Fock representation $(\HS_P,\pi_P,\Omega_P)$ of a de Sitter- and gauge-invariant state. An example for such a projection is $P=1\oplus 0$. It commutes with $\cc$ and all gauge transformations. Furthermore, if $u$ is a representation of $\tLo$ on $\hs$ of the form $u=u_1\oplus u_2$, where $u_1,u_2$ are representations of $\tLo$ on $H$ which commute with $C$ and are mutual adjoints of each other, then the associated state (\ref{eq:QuasifreeState}) is de Sitter- and gauge-invariant.

\begin{remark}
Again, we drop the $P$-dependence in our notation since all representations of the $\CAR$-algebra are faithful.
\end{remark}

In a Fock representation the gauge unitaries take the form $V(s)=e^{isQ}$, where $Q=N\otimes 1-1\otimes N$ is the charge operator and $N$ is the number operator on the Fermionic Fock space over $PH$. The Fock vacuum is invariant under gauge transformations. The spectrum of $Q$ is $\Zl$ and $\HS$ is $\Zl$-graded 
\begin{equation}
	\label{eq:FockSpaceDecomposition}
 \HS=\bigoplus_{n\in\Zl}\HS_n,\quad \HS_n=\{\Phi\in\HS:Q\Phi=n\Phi\}.
\end{equation}
The decomposition of $\HS$ into charged sectors and particle sectors are connected via
\begin{equation*}
 \HS_n=\bigoplus_{k-l=n}\wedge^kPH \otimes \wedge^lPH.
\end{equation*}
The grading is implemented by $Y=(-1)^\cal{N}$, where $\cal{N}$ is the number operator on $\HS$ (see \cite{Foit83}). In a Fock representation the (co)spinors take the form
\begin{align}
\label{eq:cospinor_Fockrep}
	\pi(\Psi(f_-))&=a^*(0\oplus PCf_-)+a(Pf_-\oplus 0)\\
 \pi(\Psi^\dagger(f_+))&=a^*(PCf_+\oplus 0)+a(0\oplus Pf_+).
\end{align}
A straightforward computation shows that the $\tLo$- and gauge-invariance of the state implies 
\begin{equation}
\label{eq:DeformationVanishesOnOmega}
 \pi(F)_{\xi,\kappa}\Omega=\pi(F)\Omega,\quad F\in\F^\infty_\xi,
\end{equation}
since two unitaries drop out in (\ref{eq:WarpedConvolution}), which yields corresponding $\delta$-factors after integration.

Now we show that the deformed and undeformed nets are unitarily inequivalent. We proceed in a similar manner as in \cite{BuchholzLechnerSummers10,DappiaggiLechnerMorfa-Morales10}. Consider the wedge $W_0$ and a rotation $r^\phi$ about an angle $\phi> 0$ in the $(x^1,x^2)$-plane. It is clear that the (bounded) region
\begin{equation*}
 \C:=r^\phi W_0\cap r^{-\phi}W_0,\quad |\phi|<\pi/2
\end{equation*}
is a subset of $W_0$ and that the reflected region ${\bj}_{W_0}\C$ lies spacelike to $W_0$ and $r^\phi W_0$.

\begin{proposition}
 Let $\F_\kappa$ be the warped convolution of the $\CAR$-net $\F$ in a Fock representation of a de Sitter- and gauge-invariant state with Reeh-Schlieder property. Suppose that $u$ is a faithful representation of $\tLo$ on $\hs$ which commutes with $C$ and $P$. Then the GNS vector $\Omega$ is not cyclic for $\F(\C)_{\kappa}\,''$ for $\kappa\neq 0$. In particular, the nets $\F$ and $\F_\kappa$ are unitarily inequivalent for $\kappa\neq 0$.
\end{proposition}
\begin{proof}
Let $f_-\in H(\C)^\infty_\xi$. Hence $f_-,u(r^{-\phi})f_-\in H(W_0)^\infty_\xi$ and the warped operators $\pi(\Psi(f_-))_{\xi,\kappa}$, $\pi(\Psi(u(r^{-\phi})f_-))_{\xi,\kappa}$ are elements of $\F(W_0)_{\kappa}$. From the de Sitter covariance (\ref{eq:DeformedOperatorsDeSitterCovariance}) follows
\begin{equation*}
 U(r^\phi)\pi(\Psi(u(r^{-\phi})f_-))_{\xi,\kappa}U(r^\phi)^{-1}
=\pi(\Psi(f_-))_{r^\phi_*\xi,\kappa},
\end{equation*}
which is an element of $\F(r^\phi W_0)_{\kappa}$.
Assume now that $\Omega$ is cyclic for $\F(\C)_{\kappa}\,''$. This is equivalent to $\Omega$ being cyclic for $\F({\bj}_{W_0}\C)_{\kappa}\,''$ since $U({\bj}_{W_0})$ is unitary and $U({\bj}_{W_0})\Omega=\Omega$. Hence $\Omega$ is separating for $\F({\bj}_{W_0}\C)_{\kappa}\,'$, which contains $\F(W_0)_{\kappa}$ and $\F(r^\phi W_0)_{\kappa}$ by locality. From (\ref{eq:DeformationVanishesOnOmega}) follows 
$\pi(\Psi(f_-))_{\xi,\kappa}\Omega
=\pi(\Psi(f_-))\Omega=\pi(\Psi(f_-))_{r^\phi_*\xi,\kappa}\Omega$, {\it i.e.}, 
\begin{equation}
\label{eq:UIsoughtcontradiction}
\pi(\Psi(f_-))_{\xi,\kappa}=\pi(\Psi(f_-))_{r^\phi_*\xi,\kappa}
\end{equation}
by the separating property of $\Omega$. Consider now a vector $\vphi\oplus 0\in PH\oplus PH$ of charge one in the one-particle space. Using Proposition \ref{prop:CARDeformation} for $m=-1$ we find
\begin{align*}
 \pi(\Psi(f_-))_{\xi,\kappa}(\vphi\oplus 0)
&=
\sum_{n\in\Zl}\pi(\Psi(u_\xi(\kappa n)f_-))U_\xi(\kappa)E(n)(\vphi\oplus 0)\\
&=
\Bl[a^*(0\oplus PCu_\xi(\kappa)f_-))+a(Pu_\xi(\kappa)f_-\oplus 0))\Br](u_\xi(\kappa)\vphi\oplus 0)\\
&=(0\oplus PCu_\xi(\kappa)f_-)\wedge \bl(u_\xi(\kappa)\vphi\oplus 0\br)+ 
(Pu_\xi(\kappa)f_-\oplus 0,u_\xi(\kappa)\vphi\oplus 0)_{\HS}\Omega\\
&=U_\xi(\kappa)(0\oplus PCf_-)\wedge(\vphi\oplus 0)
+(Pf_-,\vphi)_{PH}\Omega.
\end{align*}
For the second equality we used that $\vphi\oplus 0$ has charge one and the explicit form (\ref{eq:cospinor_Fockrep}) of cospinors in a Fock representation. For the third equality we used the usual action of the Fermi creation and annihilation operators on Fock vectors. For the fourth equality we used that $u$ commutes with $C$ and $P$ and the fact that $U_\xi(\kappa)$ is the second quantization of $u_\xi(\kappa)$. By the same computation we find 
\begin{equation*}
 \pi(\Psi(f_-))_{r^\phi_*\xi,\kappa}(\vphi\oplus 0)
=U_{r^\phi_*\xi}(\kappa)(0\oplus PCf_-)\wedge(\vphi\oplus 0)+(Pf_-,\vphi)_{PH}\Omega.
\end{equation*}
As $\Psi^\dagger(f)_{\xi,\kappa}\Phi=\Psi^\dagger(f)_{r^\phi_*\xi,\kappa}\Phi$ for all $\Phi\in\HS^\infty_\xi$ by (\ref{eq:UIsoughtcontradiction}), there follows
\begin{equation*}
 U_\xi(\kappa)(0\oplus PCf_-)\wedge(\vphi\oplus 0)
=U_{r^\phi_*\xi}(\kappa)(0\oplus PCf_-)\wedge(\vphi\oplus 0).
\end{equation*}
Since $U$ is faithful, this implies $1=\lambda_\xi(-\kappa)r^\phi\lambda_\xi(\kappa)r^{-\phi}$ which yields $\lambda_\xi(\kappa)r^\phi=r^\phi\lambda_\xi(\kappa)$. However, for $\kappa\neq 0$ this is only true for $\phi=0$ since boosts in the $x^1$-direction do not commute with rotations in the $(x^1,x^2)$-plane and contradicts our initial assumption about the rotation $r^\phi$. 

Therefore, the operator $\Psi(f_-)_{r^\phi_*\xi,\kappa}$ does depend on $\phi$, so that the cyclicity assumption on $\Omega$ for $\F(\C)_{\xi,\kappa}\,''$, $\kappa\neq 0$ is not valid. On the other hand, we know that $\Omega$ is cyclic for $\F(\C)''$ by the Reeh Schlieder property of the state. A unitary which leaves $\Omega$ invariant and maps $\F(\C)$ onto $\F(\C)_{\xi,\kappa}$ would preserve this property, from which we conclude that the undeformed and deformed net are not unitarily equivalent.
\end{proof}

\section{Conclusion and outlook} 
\label{sec:Conclusions}
In this paper we applied the warped convolution deformation method to quantum field theories with global $\Uone$ gauge symmetry on de Sitter spacetime. We used the joint action of a one-parameter group of boosts associated with a wedge and the internal symmetry group as an $\Rl^2$-action to define the deformation. The deformed theory turns out to be wedge-local and non-isomorphic to the undeformed one for a class of wedge-local field nets, including the free charged Dirac field.
\\

\noindent
In the course of writing up this paper also partial negative results were obtained, which we would like to briefly comment on. 

The warped convolution using a combination of boosts and internal symmetries can, in principle, also be defined for quantum field theories on Minkowski space. However, the covariance properties of the deformed operators are very different in this setting and a statement similar to Theorem \ref{thm:main} seems not to hold. The reason for this is that $\Gamma_{W_0}$ is, in contrast to the translations, not a normal subgroup of the Poincar\'e group. For a Poincar\'e group element $(a,\Lambda)$ one has
\begin{equation*}
 \alpha_{(a,\Lambda)}(F_{\xi,\kappa})=\alpha_{(a,\Lambda)}(F)_{(a,\Lambda)_*\xi,\kappa}
\end{equation*}
and 
\begin{equation}
\label{eq:PoincareConjugation}
 (a,\Lambda)(0,\Lambda(t))(a,\Lambda)^{-1}=(-\Lambda\Lambda(t)\Lambda^{-1}a+a,\Lambda\Lambda(t)\Lambda^{-1})
\end{equation}
is the flow which is associated with $(a,\Lambda)_*\xi$. Observe that the Lorentz group acts on $\Gamma_{W_0}$ merely by conjugation. For a translation $(a,1)W_0\subset W_0$ there is
\begin{equation*}
 \alpha_{(a,1)}(F_{\xi,\kappa})\neq \alpha_{(a,1)}(F)_{\xi,\kappa},
\end{equation*}
in general, since also a translational part is involved in (\ref{eq:PoincareConjugation}).

An interesting question is whether other Abelian subgroups of the de Sitter group can be used to define quantum field theories in terms of warped convolutions. A complete classification (up to conjugacy) of all subgroups of $\L_0$ in terms of subalgebras of its Lie algebra was given in \cite{PateraWinternitzZassehaus1976} (see also \cite{Shaw70}, \cite{Hall04} for the $\SO(1,3)_0$ case). A basis $M_{\mu\nu}$, $\mu,\nu=0,\dots,4$ for the Lie algebra of $\L_0$ satisfies
\begin{equation*}
 [M_{\mu\nu},M_{\rho\sigma}]=\eta_{\mu\rho}M_{\nu\sigma}+\eta_{\nu\sigma}M_{\mu\rho}-\eta_{\nu\rho}M_{\mu\sigma}-\eta_{\mu\sigma}M_{\nu\rho}.
\end{equation*}
It can be realized by matrices $M_{\mu\nu}=E_{\mu\nu}-E_{\nu\mu}$, $M_{0\mu}=M_{\mu 0}=E_{0\mu}-E_{\mu 0}$, where the matrix $E_{\mu\nu}$ has a one at the intersection of the $\mu$-th row and $\nu$-th column and zeros everywhere else. The two-dimensional Abelian subgroups of $\L_0$ are listed in table \ref{tab:Subgroups}.
\begin{table}[t]
\label{tab:Subgroups}
\begin{center}
	\begin{tabular}{ccc}
	\hline
	\hline
		 & subgroup & Lie algebra generators\\
	\hline
	$\L_1$ & $\SO(2)\times \SO(2)$ & $M_{12},M_{34}$\\
	$\L_2$ & $\mathrm{O}(1,1)\times \SO(2)$ & $M_{01},M_{23}$\\
	$\L_3 $& $\Rl^2$ & $M_{12}-M_{01},M_{23}-M_{03}$\\
	$\L_4$ & $\Rl\times \SO(2)$ & $M_{12}-M_{01},M_{34}$\\
	\hline
	\hline
	\end{tabular}	
\end{center}
\caption{Two-dimensional Abelian subgroups of the de Sitter group.}
\end{table}
$\L_1$ consists of spatial rotations in the $(x^1,x^2)$- and $(x^3,x^4)$-plane. $\L_2$ are boosts in the $x^1$-direction and rotations in the $(x^2,x^3)$-plane. $\L_3$ corresponds to null rotations (translational part of the stabilizer group of a light ray). $\L_4$ is a combination of a null rotation and a spatial rotation. All of these groups can be used to define a warped convolution with the associated $\Rl^2$-action from the representation. 
Since we are on a curved spacetime it appears to be reasonable to require that the group which is used for the deformation is a subgroup of the stabilizer of a wedge. The reason is that there is not an analogue of the spectrum condition on Minkowski space available which restricts the spectral properties of the generators which are associated with isometries (The microlocal spectrum condition only gives a restriction on the singularity structure of the two-point function.). Comparing the subgroup structure of $\L_0(W_0)$ with the above groups shows that only $\L_2$ is a subgroup. However, $\L_2$ violates conditions a) and b) in Definition \ref{def:CausalBorchersSystem} for certain reflections: Denote by $F_{\zeta,\kappa}$ the warped operator, where $\zeta=(M_{01},M_{23})$ is a pair of Killing vector fields (compare \cite{DappiaggiLechnerMorfa-Morales10}) and consider the reflection $j_{12}(x^0,x^1,x^2,x^3,x^4)=(x^0,-x^1,-x^2,x^3,x^4)$ which satisfies $j_{12}W_0=(W_0)'$. The associated flow $\Lambda(t,s):=\exp(t M_{01})\exp(s M_{23}),\,  t,s\in\Rl$ transforms as
\begin{equation*}
 j_{12}\Lambda(t,s)j_{12}=\Lambda(-t,-s)
\end{equation*}
so that $\alpha_{j_{12}}(F_{\zeta,\kappa})=\alpha_{j_{12}}(F)_{{j_{12}}_*\zeta,\kappa}=F_{\zeta,\kappa}$ and condition b) is violated. Similar problems also appear if one uses a combination of boosts and translations along the edge of the wedge. From these observations we conclude that the position of the subgroup, which is used for the deformation, within the isometry is very important and that a modification of the standard warping formula is necessary in these cases. 

A deformation with purely internal symmetries, e.g. $\Uone\times\Uone$, did not appear to be interesting, because an adaption of Proposition \ref{prop:CARDeformation} to this case yields that the deformation is trivial on the level observables and also trivial for generators $B(f)$, provided the induced charge structure of the gauge groups is the same.
\\

\noindent
The deformation scheme in this paper is very different from the one in \cite{DappiaggiLechnerMorfa-Morales10}, where the Killing flow associated with the edge of a wedge was used to formulate the deformation. It would be desirable to establish a connection between the two approaches. In the de Sitter case the edge is a two-sphere, which is an $\SO(3)$-orbit, and an approach as in \cite{DappiaggiLechnerMorfa-Morales10} would require a generalization of the warped convolution to group actions of $\SO(3)$. But deformations of C$^*$-algebras which are based on actions of general non-Abelian groups do not seem to be available so far (see however \cite{Bieliavsky2002} for certain examples).

It appears to be a challenging task to generate new examples of deformed quantum field theories whose covariance and localization properties are well-behaved. However, the quest of finding new deformation formulas is a worthwhile task and is expected to yield better insights into the nature of interacting quantum field theories.
\\

\noindent
{\bf Acknowledgements}\\
I would like to thank M. Bischoff, C. Dappiaggi, W. Dybalski, H. Grosse, G. Lechner, J. Schlemmer and J. Yngvason for helpful discussions. I am grateful to G. Lechner for useful comments on a first draft of this paper. The hospitality of the Erwin Schr\"odinger Institute and the financial support of the EU-network MRTN-CT-2006-031962, EU-NCG are gratefully acknowledged.

%
%

\end{document}